\documentclass[copyright,creativecommons]{eptcs}
\providecommand{\event}{Termgraph 2014}
\usepackage[all]{xy}
\usepackage{datetime}
\renewcommand{\tt}{\rm \ttfamily \bfseries }
\usepackage[normalem]{ulem}

\newcommand{\codefont}{\tt}
\newcommand{\equprogram}[1]{%
\def\separator{1.3333ex}%
\frenchspacing%
\refstepcounter{equation}%
\par\vspace\separator\hspace{1em}%
$\vcenter{\codefont\noindent{#1}}$%
\raisebox{-0.3ex}{\kern-2.75em\llap{\rm (\theequation)}}
\par\vspace\separator\noindent\kern-.0em%
}
\newcommand{\code}[1]{\mbox{\codefont{#1}}}
\newcommand{\us}{\raise-.8ex\hbox{-}}
\newcommand{\xtilde}{$\!$\raise-.75ex\hbox{\char`\~}} 

\newcommand{\Signature}{\Sigma}
\newcommand{\Variables}{{\cal X}}
\newcommand{\Rules}{{\cal R}}
\newcommand{\Constructors}{{\cal C}}
\newcommand{\Operations}{{\cal D}}

\newcommand{\tostar}{\buildrel * \over \to}
\newcommand{\toplus}{\buildrel \mbox{\rm \tiny +} \over \to}
\newcommand{\toequal}{\to^{\llap{\raisebox{.4ex}{\tiny =~}}}}

\newcommand{\norm}{\mbox{\bf N}}
\newcommand{\head}{\mbox{\bf H}}

\def\tree{{\cal T}}
\def\pattern{\pi}

\newcommand{\w}{\noindent{\rm \phantom{XX}}} 
\newcommand{\interpol}[1]{\$\{#1\}}
\newcommand{\erasure}{{\cal E}}
\newcommand{\headcomp}[1]{\mbox{\bf H$_{#1}$}}

\usepackage{amsthm}
\newtheorem{definition}{Definition}
\newtheorem{lemma}{Lemma}
\newtheorem{theorem}{Theorem}
\newtheorem{corollary}{Corollary}
\newtheorem{example}{Example}

\newcommand{\COMMENT}[1]{}
\def\bigbox#1{\raisebox{-3.5pt}{\vbox{\hrule\hbox{\vrule\kern8pt
  \vbox{\kern8pt\hbox{#1}\kern8pt}\kern8pt\vrule}\hrule}}}
\def\smallbox#1{\raisebox{-3.5pt}{\vbox{\hrule\hbox{\vrule\kern3pt
  \vbox{\kern3pt\hbox{#1}\kern3pt}\kern3pt\vrule}\hrule}}}

\title{Needed Computations Shortcutting Needed Steps}

\author{Sergio Antoy
\institute{Computer Science Dept.\\ Portland State University\\ Oregon, U.S.A.}
\email{antoys@pdx.edu}
\and
Jacob Johannsen
\institute{Dept. of Computer Science\\ Aarhus University\\ Denmark }
\email{cnn@cs.au.dk}
\and
Steven Libby
\institute{Computer Science Dept.\\ Portland State University\\ Oregon, U.S.A.}
\email{slibby@pdx.edu}
}
\def\titlerunning{Needed Computations Shortcutting Needed Steps}
\def\authorrunning{S. Antoy, J. Johannsen, \& S. Libby}

\begin{document}
\maketitle

\begin{abstract}
We define a compilation scheme for a constructor-based, strongly-sequential,
graph rewriting system which shortcuts some needed steps.
The object code 
is another constructor-based graph rewriting system.
This system is normalizing for the original system
when using an innermost strategy.
Consequently, the object code can be easily implemented by eager
functions in a variety of programming languages.
We modify this object code in a way that 
avoids total or partial construction of the \emph{contracta}
of some needed steps of a computation.
When computing normal forms in this way,
both memory consumption and execution time are reduced compared
to ordinary rewriting computations in the original system.
\end{abstract}


\section{Introduction}
\label{Introduction}

Rewrite systems are models of computations
that specify the actions, but not the control.
The object of a computation is a graph referred to as an
\emph{expression}.  
The actions are encoded by rules that define how
to replace (rewrite) one expression with another. 
The goal of a computation is to reach an expression,
called a \emph{normal form}, that cannot be further rewritten.

In the computation of an expression, 
a rewrite system does not tell which subexpression should be replaced
to reach the goal.
\begin{example}
Consider the following rewrite system.
The syntax is Curry \cite{Hanus12Curry}.
\equprogram{%
\label{contrived}%
loop = loop \\
snd (\us,y) = y
}
A computation of \code{snd\,(loop,0)} terminates with \code{0} if
the second rule of (\ref{contrived}) is ever applied, but
goes on forever without making any progress if only the first
rule is applied.
\end{example}
\noindent
In a computation a \emph{strategy} is a policy or algorithm that
defines both which subexpression should be replaced
and its replacement.
The intended goal of a strategy is to efficiently produce a
normal form of an expression when it exists.
A practical strategy, called \emph{needed},
is known for the class of the strongly sequential
term rewriting systems \cite{HuetLevy91-both}.
This strategy relies on the fact that, in every reducible expression $e$,
there exists a \emph{redex},
also called \emph{needed}, that is
reduced in any computation of $e$ to a normal form.

The needed strategy is defined and implemented as follows:
given an expression $e$, while $e$ is reducible, reduce an arbitrarily
chosen, needed redex of $e$.  In the systems considered in
this paper, finding a needed redex is
easy without look-ahead \cite{Antoy05JSC}.
This strategy is normalizing:
if an expression $e$ has a normal form, repeatedly reducing
arbitrary needed redexes will terminate with that normal form.
This strategy is also optimal in the number
of reduced redexes for graph (not term) rewriting.

The above outline shows that implementing a needed strategy is a
relatively straightforward task. Surprisingly, however, it is possible
to shortcut some of the needed steps in the computation. This paper
shows how this shortcutting can be introduced into an implementation
of a needed strategy.

Terminology and background information are
recalled in Sect.~\ref{Preliminaries}.
The compilation scheme and its properties are
in Sect.~\ref{Compilation} and \ref{Compilation Properties}.
The transformation that allows shortcutting needed redexes,
and its properties, are in Sect.~\ref{Transformation}
and \ref{Transformation Properties}.
Sect.~\ref{Benchmarking} presents two benchmarks
and sketches further opportunities
to shortcut needed steps.
Sect.~\ref{Functional Logic Programming} discusses the application
of our work to the implementation of functional logic languages.
Related work and conclusion are in Sect.~\ref{Related Work}
and \ref{Conclusion}, respectively.

\section{Preliminaries}
\label{Preliminaries}

A \emph{rewrite system} is a pair $(\Signature \cup \Variables, \Rules)$
in which
$\Signature = \Constructors \uplus \Operations$
is a \emph{signature} partitioned into 
\emph{constructors} and \emph{operations} (or functions), 
$\Variables$ is a denumerable set of \emph{variables},
and $\Rules$ is a set of \emph{rewrite rules} defined below.
Without further mention, we assume that the symbols of
the signature have a type, and that any expression
over the signature is well typed.

An \emph{expression} is a single-rooted, directed, acyclic \emph{graph}
defined in the customary way~\cite[Def. 2]{EchahedJanodet97IMAG}.
An expression $e$ is a \emph{constructor form} (or \emph{value})
if, and only if, every node of $e$ is labeled by a constructor symbol.
Constructor forms are normal forms, but not vice versa.
For example, $\emph{head}([\,])$ where \emph{head} is the usual operation that
returns the first element of a (non-empty) list, is a normal form, but not a constructor form.
In a constructor-based system,
such expressions are regarded as \emph{failures} or 
\emph{exceptions} rather than results of computations.
Likewise, a \emph{head-constructor form} is an expression
whose root node is labeled by a constructor symbol.

A \emph{rule} is a graph with two roots abstracting
the left- and right-hand sides respectively.
The rules follow the \emph{constructor discipline} \cite{ODonnell77}.  
Each rule's left-hand side is a \emph{pattern}, i.e.,
an \emph{operation} symbol applied to zero or more expressions
consisting of only \emph{constructor} symbols and variables.
Rules are left-linear, i.e., the left-hand side is a tree.

The objects of a computation are graphs rather than terms.
Sharing some subexpressions of an expression
is a requirement of functional logic programming
\cite{CRWL99JLP,Hanus94JLP,Hanus13,LopezFraguas14TPLP}. 
Incidentally, this sharing ensures that needed redexes
are never duplicated during a computation.
The difference between our graphs and ordinary terms 
concerns only the sharing of subexpressions.

A \emph{computation} of an expression $t$ is a possibly infinite sequence 
\begin{displaymath}
t = t_0 \to t_1 \to \ldots
\end{displaymath}
such that
$t_i \to t_{i+1}$ is a rewrite step \cite[Def. 23]{EchahedJanodet97IMAG}.
For all $i$, $t_i$
is a \emph{state} of the computation of $t$.

Given a rewrite system $R$, an \emph{expression of} $R$
is an expression over the signature of $R$.
When $s$ is a signature symbol and $n$ is a natural number,
$s/n$ denotes that $n$ is the \emph{arity} of $s$.
When $t$ and $u$ are expressions and $v$ is a variable,
$[u/v]$ is the \emph{substitution} that maps $v$ to $u$, and
$t[u/v]$ is the application of $[u/v]$ to $t$.
The reflexive closure of the rewrite
relation ``$\to$'' is denoted ``$\toequal$''.

Each operation in $\Operations$ is \emph{inductively sequential};
that is, its rewrite rules are organized into a hierarchical structure called
a \emph{definitional tree} \cite{Antoy92ALP} which we informally
recall below.  
An example of a definitional tree is shown in (\ref{def-tree-example}).
In a definitional tree of an operation $f$, there are up to 3 kinds of nodes
called \emph{branch}, \emph{rule} and \emph{exempt}.
Each kind contains a pattern of $f$ and other items of information
depending on the kind.
A \emph{rule} node with pattern $\pattern$ contains a rule of $f$
whose left-hand side is equal to $\pattern$
modulo renaming nodes and variables.
An \emph{exempt} node with pattern $\pattern$ contains no other information.  
There is no rule of $f$ whose left-hand side is equal to $\pattern$.
A \emph{branch} node with a pattern $\pattern$ contains
children that are subtrees of the definitional tree.
At least one child is a \emph{rule} node.
The children are obtained by ``narrowing'' pattern $\pattern$.
Let $x$ be any variable of $\pattern$, which is called \emph{inductive}.
For each constructor $c/m$ of the type of $x$,
there is a child whose pattern is obtained from $\pattern$
by instantiating $x$ with $c(x_1,\ldots x_m)$, where $x_i$ is a fresh
variable.  An operation $f/n$ is \emph{inductively sequential}
\cite{Antoy92ALP}
if there exists a definitional tree whose root has pattern
$f(x_1,\ldots x_n)$, where $x_i$ is a fresh variable,
and whose leaves contain all, and only, the rules of $f$.
A rewrite system is \emph{inductively sequential} if
all of its operations are inductively sequential.

Inductively sequential operations can be thought of as ``well designed'' 
with respect to evaluation.
To compute a needed redex of an expression $e$ rooted by an operation $f$,
match to $e$ the pattern $\pattern$ of a maximal (deepest in the tree) 
node $N$ of a definitional tree of $f$.
If $N$ is an \emph{exempt} node, $e$ has no constructor normal form,
and the computation can be aborted.
If $N$ is a \emph{rule} node, $e$ is a redex and can be
reduced by the rule in $N$.
If $N$ is a \emph{branch} node, let $x$ be the inductive variable
of $\pattern$ and $t$ the subexpression of $e$ to which
$x$ is matched.
Then, recursively compute a needed redex of $t$.

The inductively sequential systems are the intersection
\cite{HanusLucasMiddeldorp98IPL} of the strongly sequential systems
\cite{HuetLevy91-both} and the constructor-based systems
\cite{ODonnell77}.
The following notion \cite{AntoyJost13PPDP}
for inductively sequential systems is key to our work.
We abuse the word ``needed'' because we will show that our notion
extends the classic one \cite{HuetLevy91-both}.
Our notion is a binary relation on nodes,
or equivalently on the subexpressions rooted by these nodes,
since they are in a bijection.

\begin{definition}
\label{needed-node}
Let $R$ be an inductively sequential system,
$e$ an expression of $R$ rooted by a node $p$,
and $n$ a node of $e$.
Node $n$ is \emph{needed for} $e$, and similarly 
is \emph{needed for} $p$,
if, and only if, in any computation of $e$
to a head-constructor form, the subexpression of $e$ at
$n$ is derived to a head-constructor form.
A node $n$ (and the redex rooted by $n$, if any)
of a state $e$ of a computation in $R$ is \emph{needed} if, and only if, it is needed
for some outermost operation-rooted subexpression of $e$.
\end{definition}
Our ``needed'' relation is interesting only when both nodes
are labeled by operation symbols.
If $e$ is an expression whose root node $p$ is labeled
by an operation symbol, then $p$ is trivially needed for $p$.
This holds whether or not $e$ is a redex and even
when $e$ is \emph{already} a normal form, e.g., $\emph{head}([\,])$.
In particular, \emph{any} expression that is not a value has 
pairs of nodes in the needed relation.
Finally, our definition is concerned with
reaching a \emph{head-constructor} form,
not a \emph{normal form}.

Our notion of need generalizes the classic notion
\cite{HuetLevy91-both}.  Also, since our systems follow the constructor
discipline \cite{ODonnell77} we are not interested 
in expressions that do not have a value.

\begin{lemma}
\label{extension}
Let  $R$ be an inductively sequential system and $e$ an expression
of $R$ derivable to a value.
If $e'$ is an outermost operation-rooted
subexpression of $e$, and $n$ is both a node needed for $e'$ and
the root of a redex $r$, then $r$ is a needed redex of $e$
in the sense of \emph{\cite{HuetLevy91-both}}.
\end{lemma}
\begin{proof}
Since $e'$ is an outermost operation-rooted subexpression of $e$,
any node in any path from the root of $e$ to the root of $e'$,
except for the root of $e'$, is labeled by constructor symbols.
Hence, $e$ can be derived to a value only if $e'$ is derived to
a value and $e'$ can be derived to a value only if
$e'$ is derived to a head-constructor form.
By assumption, in any derivation of $e'$ to a head-constructor form
$r$ is derived to a head-constructor form,
hence it is reduced.  Thus, $r$ is a needed redex of $e$
according to \cite{HuetLevy91-both}.
\end{proof}

\begin{lemma}
\label{need-transitivity}
Let  $R$ be an inductively sequential system, $e$ an expression
of $R$, $e_1$, $e_2$ and $e_3$ subexpressions of $e$
such that $n_i$ is the root of $e_i$ and the label
of $n_i$ is an operation, for $i=1,2,3$.
If $n_3$ is needed for $n_2$ and $n_2$ is needed for $n_1$,
then $n_3$ is needed for $n_1$.
\end{lemma}
\begin{proof}
By hypothesis,
if $e_3$ is not derived to a constructor-rooted form,
$e_2$ cannot be derived to a constructor-rooted form,
and 
if $e_2$ is not derived to a constructor-rooted form,
$e_1$ cannot be derived to a constructor-rooted form.
Thus,
if $e_3$ is not derived to a constructor-rooted form,
$e_1$ cannot be derived to a constructor-rooted form.
\end{proof}

\section{Compilation}
\label{Compilation}

\begin{figure}
\label{compiler}%
\bigbox{%
$\vcenter{%
{\tt
compile $\tree$ \\
{\tiny 01}\w case $\tree$ of \\
{\tiny 02}\w when $branch(\pattern,o,\bar \tree)$ then \\
{\tiny 03}\w\w $\forall\, \tree_i \in \bar \tree$ compile $\tree_i$ \\
{\tiny 04}\w\w output \mbox{\rm ``}$\head(\interpol{\pattern}) = 
   \head(\interpol{\pattern[\head(\pattern|_o)]_o})$\mbox{\rm ''} \\
{\tiny 05}\w when $rule(\pattern,l \to r)$ then \\
{\tiny 06}\w\w case $r$ of \\
{\tiny 07}\w\w\w when operation-rooted then \\
{\tiny 08}\w\w\w\w output \mbox{\rm ``$\head(\interpol{l}) = \head(\interpol{r})$''}  \\
{\tiny 09}\w\w\w when constructor-rooted then \\
{\tiny 10}\w\w\w\w output \mbox{\rm ``$\head(\interpol{l}) = \interpol{r}$''}  \\
{\tiny 11}\w\w\w when variable then \\
{\tiny 12}\w\w\w\w for each constructor $c/n$ of the sort of $r$\\
{\tiny 13}\w\w\w\w\w let $l' \to r' = (l \to r)[c(x_1,\ldots x_n)/r]$ \\
{\tiny 14}\w\w\w\w\w output \mbox{\rm ``$\head(\interpol{l'}) = \interpol{r'}$''} \\
{\tiny 15}\w\w\w\w output \mbox{\rm ``$\head(\interpol{l}) = \head(\interpol{r})$''}  \\
{\tiny 16}\w when $exempt(\pattern)$ then \\
{\tiny 17}\w\w output \mbox{\rm ``}$\head(\interpol{\pattern}) = \mbox{\rm abort}$\mbox{\rm ''}
}
}$ 
 \kern-22pt } 
\caption{\label{compile}
Procedure \code{compile} takes a definitional tree of an operation $f$
of $R$ and produces the set of rules of $\head$ that pattern match
$f$-rooted expressions.}
\end{figure}

For simplicity and abstraction,
we present the object code, $C_R$, of $R$
as a constructor-based graph rewriting system
as well.
$C_R$ has only two operations called \emph{head} and \emph{norm},
and denoted $\head$ and $\norm$, respectively.
The constructor symbols of $C_R$ are all, and only, the
symbols of $R$.
The rules of $C_R$ have a priority established by the textual order.
A rule reduces an expression $t$ only if no other preceding rule
could be applied to reduce $t$.
These semantics are subtle, since $t$ could become reducible by the
preceding rule only after some internal reduction.
However, all our claims about computations in $C_R$ are stated for
an innermost strategy.
In this case, when a rule is applied, no internal reduction
is possible, and the semantics of the priority are straightforward.

Operation $\head$ is defined piecemeal for each operation of $R$.
Each operation of $R$ contributes a number of rules 
dispatched by pattern matching with textual priority.
The rules of $\head$ contributed by an operation with definitional
tree $\tree$ are generated by the procedure \code{compile} defined 
in Fig.~\ref{compile}.
The intent of $\head$ is to take an expression of $R$ rooted by
an operation and derive an expression
of $R$ rooted by a constructor by performing only needed steps. 

The expression ``$\interpol{x}$'' embedded in a string, denotes
interpolation as in modern programming languages, i.e., the argument $x$
is replaced by a string representation of its value.
The notation $t[u]_p$ stands for an expression
equal to $t$, in which the subexpression identified by $p$ is replaced by $u$.
In procedure \code{compile}, the notation is used
to ``wrap'' an application of $\head$ around
the subexpression of the pattern at $o$,
the inductive node.
An example is the last rule of (\ref{append-compiled}).
The loop at statement 12 is for collapsing rules,
i.e., rules whose right-hand side is a variable.
When this variable matches an expression rooted by a constructor
of $R$, no further application of $\head$ is required;
Otherwise, $\head$ is applied to the contractum.
Symbol ``abort'' is not considered an element of the signature
of $C_R$.  If any redex is reduced to ``abort'', the computation is
aborted since it can be proved that the expression object of the
computation has no \emph{constructor} normal form.

\begin{example}
Consider the rules defining the operation that
concatenates two lists, denoted by the infix identifier ``\code{++}'':
\equprogram{%
\label{append-rules}%
[]++y = y \\
(x:xs)++y = x:(xs++y)
}
The definitional tree of operation ``\code{++}'' is pictorially
represented below.
The only \emph{branch} node of this tree is the root.
The inductive variable of this branch, boxed in the representation, is \code{x}.
The \emph{rule} nodes of this tree are the two leaves.
There are no \emph{exempt} nodes in this tree since 
operation ``\code{++}'' is completely defined.
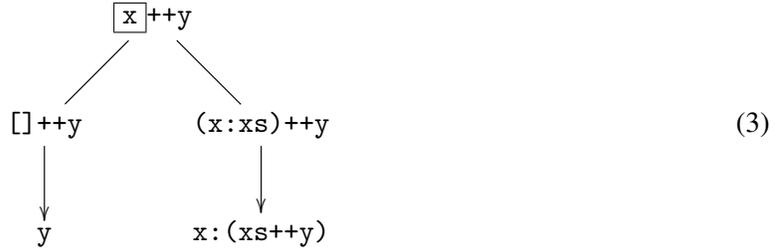
\begin{equation}
\label{def-tree-example}
\vcenter{
\xymatrix@!=17pt{
& \mbox{\tt \smallbox{x}++y}
       \ar@{-}[dl]\ar@{-}[dr] \\
  \mbox{\tt []++y} \ar[d] 
   & & \mbox{\tt (x:xs)++y}  \ar[d]  \\
  \mbox{\tt y} & & \mbox{\tt x:(xs++y)}
}
}
\end{equation}
\noindent
Applying procedure \code{compile} to this tree produces the following output:
\newcommand{\numline}[1]{\hfill \llap{\makebox[6in]{\rm compile line \##1}}}
\equprogram{%
\label{append-compiled}%
\head([]++[]) = []                \numline{14}\\
\head([]++(y:ys)) = (y:ys)        \numline{14}\\
\head([]++y) = \head(y)           \numline{15}\\
\head((x:xs)++y) = x:(xs++y)      \numline{10}\\
\head(x++y) = \head(\head(x)++y)  \numline{04}
}
\end{example}
\noindent
Operation $\norm$ of the object code is defined by one rule
for each symbol of $R$.  In the following \emph{metarules},
$c/m$ stands for a constructor of $R$,  
$f/n$ stands for an operation of $R$,
and $x_i$ is a fresh variable for every $i$.
\equprogram{%
\label{function-norm}%
$\norm(c(x_1,\ldots x_m)) = c(\norm(x_1), \ldots \norm(x_m))$ \\
$\norm(f(x_1,\ldots x_n)) = \norm(\head(f(x_1,\ldots x_n)))$
}
\begin{example}
The rules of $\norm$ for the list constructors
and the operation ``\code{++}'' defined earlier are:
\equprogram{%
\label{norm-example}%
\norm([]) = [] \\
\norm(x:xs) = \norm(x):\norm(xs) \\
\norm(x++y) = \norm(\head (x++y))
}
\end{example}
\begin{definition}
The rewrite system consisting of the $\head$
rules generated by procedure
\code{compile} for the operations of $R$ and 
the $\norm$ rules generated according to (\ref{function-norm})
for all the symbols of $R$ is the
\emph{object code} of $R$ and is denoted $C_R$.
\end{definition}

\begin{example}
We show the computation of \code{[1]++[2]}
in both $R$ and $C_R$. We use the desugared notation for list 
expressions to more easily match the patterns of the rules of ``\code{++}''.
\equprogram{%
\label{comp-in-R}%
(1:[])++(2:[]) $\to$ 1:([]++(2:[])) $\to$ 1:2:[]
}
and
\equprogram{%
\label{comp-in-CR}%
\norm((1:[])++(2:[])) \\
\w $\to$   \norm(\head((1:[])++(2:[]))) \\
\w $\to$   \norm(1:([]++(2:[])) \\
\w $\to$   \norm(1):\norm([]++(2:[])) \\
\w $\to$   1:\norm(\head([]++(2:[]))) \\
\w $\to$   1:\norm(2:[]) \\
\w $\to$   1:\norm(2):\norm([]) \\
\w $\to$   1:2:[]
}
Computation (\ref{comp-in-CR}) is longer than (\ref{comp-in-R}).
If all the occurrences
of $\norm$ and $\head$ are ''erased'' from the states of (\ref{comp-in-CR}),
a concept formalized shortly, and repeated states of the computation
are removed, the remaining steps are the same as in (\ref{comp-in-R}).
The introduction and removal of occurrences
of $\norm$ and $\head$ in (\ref{comp-in-CR}),
which lengthen the computation,
represent the control, what to rewrite and when to stop.
These activities occur in (\ref{comp-in-R}) too,
but are in the mind of the reader
rather than explicitly represented in the computation.
\end{example}

\section{Compilation Properties}
\label{Compilation Properties}

$C_R$, the object code of $R$, correctly implements $R$.
Computations performed
by $C_R$ produce the results of corresponding computations in $R$
as formalized below.
Furthermore, $C_R$ implements a needed strategy, because every reduction
performed by $C_R$ is a needed reduction in $R$. In this section, we
prove these properties of the object code.

Let $Expr$ be the set of expressions over
the signature of $C_R$ output by \code{compile}
on input a rewrite system $R$.
The \emph{erasure} function $\erasure : Expr \to Expr$
is inductively defined by:
\equprogram{%
\label{erasure}%
$\erasure(H(t)) = \erasure(t)$ \\
$\erasure(N(t)) = \erasure(t)$ \\
$\erasure(s(t_1,\ldots t_n)) = s(\erasure(t_1),\ldots \erasure(t_n))$ \mbox{\rm for} $s/n \in \Sigma_R$
}
Intuitively, the erasure of an expression $t$ removes all the occurrences
of $\head$ and $\norm$ from $t$.
The result is an expression over the signature of $R$.

\begin{lemma}
\label{completeness}%
Let $R$ be an inductively sequential system and
$\head$ the head function of $C_R$.
For any operation-rooted expression $t$ of $R$, $\head(t)$ is a redex.
\end{lemma}
\begin{proof}
Let $f/n$ be the root of $t$, and
$\tree$ the definitional tree of $f$ input to procedure \code{compile}.
The pattern at the root of $\tree$ is $f(x_1,\ldots x_n)$,
where each $x_i$ is a variable.
Procedure \code{compile} outputs a rule of $\head$ with left-hand side
$\head(f(x_1,\ldots x_n))$.  Hence this rule, or
a more specific one, reduces $t$.
\end{proof}
\noindent
Comparing graphs modulo a renaming of nodes,
as in the next proof,
is a standard technique \cite{EchahedJanodet97IMAG}
due to the fact that any node created by a rewrite is fresh.

\begin{lemma}
\label{operation-rooted}%
Let $R$ be an inductively sequential system and
$\head$ the head function of $C_R$.
Let $t$ be an operation-rooted expression of $R$, and
$\head(t)$ be reduced by a step
resulting from the application of a rule $r$ originating
from statement 04 of procedure \code{compile}.
The argument of the inner application of $\head$ in the contractum
is both operation-rooted and needed for $t$.
\end{lemma}
\begin{proof}
Let $\tree$ be a definitional tree of the root of $t$.
Let $\pattern$ be the pattern of the \emph{branch} node $n$
of $\tree$ from which rule $r$ originates and let $o$ be the
inductive node of $\pattern$.
Since $r$ rewrites $t$ and $\pattern$ is the left-hand side of $r$
modulo a renaming of variables and nodes,
there exists a graph homomorphism $\sigma$ such that
$t = \sigma(\pattern)$.
Our convention on the specificity of the rules defining $\head$
establishes that no rule textually preceding $r$ in the 
definition of $\head$ rewrites $t$.
Since procedure \code{compile} traverses $\tree$ in post-order,
every rule of $\head$ originating from a node descendant of $n$
in $\tree$ textually precedes $r$ in the definition of $\head$.
Let $q = \sigma(\pattern|_o)$.
For each constructor symbol $c/n$ of $R$ of the sort of $\pattern|_o$,
there is a rule of $\head$ with argument 
$\pattern[c(x_1,\ldots,x_n)]|_o$, where $x_1,\ldots,x_n$ are fresh variables,
and this rule textually precedes $r$ in the definition of $\head$.
Therefore, the label of $q$ is not a constructor symbol, otherwise
this rule would be applied to $t$ instead of $r$.
Since the step of $t$ is innermost, $q$ cannot be labeled by $\head$ either.
Thus, the only remaining possibility is that $q$ is labeled by an operation.
We now prove that $q$ is needed for $t$.
If $n_1$ and $n_2$ are disjoint nodes 
(neither is an ancestor of the other) of $\tree$,
then the patterns of $n_1$ and $n_2$ are not unifiable.
This is because they have different constructors symbols
at the node of the inductive variable of the closest
(deepest) common ancestor.
Thus, since $t=\sigma(\pattern)$, only a rule of $R$
stored in a \emph{rule} node of $\tree$ below $n$
can rewrite (a descendant of) $t$ at the root, if any such a rule exists.
All these rules have a constructor symbol at the node matched by $o$,
whereas $t$ has an operation symbol at $q$, the node matched by $o$.
Therefore, $t$ cannot be reduced (hence reduced to a head-constructor
form) unless $t|_q$ is reduced to a head-constructor form.
Thus, $q$ is needed for $t$. 
\end{proof}
\begin{example}
The situation depicted by the previous lemma can be seen in the
evaluation of \linebreak
$t = \code{([1]++[2])++[3]}$.
According to (\ref{append-compiled}),
$\head(t) \to \code{\head(\head([1]++[2])++[3])}$.
The argument of the inner application of $\head$ is 
both operation-rooted and needed for $t$.
\end{example}

\begin{lemma}
\label{simulated}%
Let $R$ be an inductively sequential system and
$\head$ the head function of $C_R$.
Let $t$ be an operation-rooted expression of $R$ and let $A$ denote
an innermost finite or infinite 
computation $\head(t) = e_0 \to e_1 \to \ldots$ in $C_R$.  
\begin{enumerate}
\item [\rm 1.]
For every index $i$ in $A$, $\erasure(e_i) \toequal \erasure(e_{i+1})$ in $R$.
\item [\rm 2.]
If $A$ terminates (it neither aborts nor is infinite)
in an expression $u$, 
then  $u$ is a head-constructor form of $R$.
\end{enumerate}
\end{lemma}
\begin{proof}
Claim 1:
Let $l \to r$ be the rule of $\head$ applied in the step $e_i \to e_{i+1}$.
There are 3 cases for the origin of $l \to r$.
If $l \to r$ originates from statement 04 of \code{compile},
then $\erasure(e_i) = \erasure(e_{i+1})$ and the claim holds.
Otherwise $l \to r$ originates from one of statements 08, 10, 14 or 15.
In all these cases, a subexpression of $e_i$ of the form $\head(w)$
is replaced by either $\head(u)$ (statements 08 and 15)
or $u$ (statements 10 and 14),
in which $w$ is an instance of the left-hand side of a rule of $R$
and $u$ is the corresponding right-hand side.  Thus, in this
case too, the claim holds.
\\
\COMMENT{Consider the following computation 
\equprogram{
$\head($([1]++[2])++[3]$)$ \\
$~~\to \head(\head($([1]++[2])++[3]$)$ \\
$~~\to \head(\head(\head($([1]++[2])++[3]$)$ \\
$~~\to \head(\head($(1:([]++[2]))++[3]$)$ 
}
Now, $\head$ of a cons is unintended and no constructor form is reachable. \\
If reductions are innermost, the problem goes away.
}
Claim 2:
If $A$ aborts or does not terminate, the claim is vacuously true.
So, consider the last step of $A$.
This step cannot originate from the application of a rule
that places $\head$ at the root of the contractum,
since another step would become available.
Hence the rule of the last step is generated by statement 10 or
14 of procedure \code{compile}.
In both cases, the contractum is a head constructor form.
\end{proof}
\noindent
If $A$ denotes a computation $\head(t) = e_0 \to e_1 \to \ldots$ in $C_R$,
then, by Lemma \ref{simulated}, 
we denote $\erasure(e_0) \toequal \erasure(e_1) \toequal \ldots$
with $\erasure(A)$ and---with a slight abuse---we regard
it as a computation in $R$.
Some expression of $\erasure(A)$ may be a repetition of the previous one,
rather than the result of a rewrite step.  However, it is more practical to
silently ignore these duplicates than filtering them out at
the expenses of a more complicated definition.
We will be careful to avoid an infinite repetition of the same expression.
We extend the above viewpoint to computations of $\norm(t)$ in $C_R$,
where $t$ is any expression of $R$.

\begin{theorem}
\label{main}
Let $R$ be an inductively sequential system and
$\head$ the head function of $C_R$.
Let $t$ be an operation-rooted expression of $R$ and let $A$ denote
an innermost finite or infinite 
computation $\head(t) = e_0 \to e_1 \to \ldots$ in $C_R$.
Every step of $\erasure(A)$ is needed.
\end{theorem}
\begin{proof}
We prove that for every index $i$ such that $e_i$ is a state
of $A$, every argument of an
application of $\head$ in $e_i$ is needed for $\erasure(e_i)$.
Preliminarily, we define a relation ``$\prec$'' on the nodes of the
states of $\erasure(A)$ as follows. 
Let $p$ and $q$ be nodes of 
states $\erasure(e_i)$ and $\erasure(e_j)$ of $\erasure(A)$ respectively. We
define $p \prec q$ iff $i<j$ or $i=j$ and the expression at $q$
is a proper subexpression of the expression at $p$ in $\erasure(e_i)$.
Relation ``$\prec$'' is a well-founded ordering 
with minimum element the root of $t$.
The proof of the theorem is by induction on ``$\prec$''.
Base case: Directly from the definition of ``need'',
since $t$ is rooted by an operation of $R$.
Induction case:
Let $q$ be the root of the argument of an application
of $\head$ in $e_j$ for $j>0$.
We distinguish whether $q$ is the root
of the argument of an application of $\head$ in $e_{j-1}$.
If it is, then the claim is a direct consequence of the
induction hypothesis.
If it is not, $e_{j-1} \to e_j$ is an application of a rule $r$ generated
by one of the statements 04, 08 or 15 of procedure \code{compile}.
For statement 04, there is a node $p$ of $\erasure(e_j)$ that
by the induction hypothesis is needed for $\erasure(e_j)$
and matches the pattern $\pattern$
of the branch node of a definitional tree from which rule $r$ originates.
Let $q$ be the node of the subexpression of $e_j$ rooted by $p$
matched by $\pattern$ at $o$.
By Lemma \ref{operation-rooted}, $q$ is needed for $p$.
Since $p$ is needed for $\erasure(e_j)$, by Lemma \ref{need-transitivity},
$q$ is needed for $\erasure(e_j)$ and the claim holds.
For statements 08 and 15, $q$ is the root of the contractum of the
redex matched by $r$ which by the induction hypothesis
is needed for $\erasure(e_{j-1})$.
Node $q$ is still labeled by an operation, hence it is needed 
for $\erasure(e_j)$ directly by the definition of ``need''.
\end{proof}

\begin{corollary}
Let $R$ be an inductively sequential system.
Let $t$ be an expression of $R$ and let $A$ denote
an innermost finite or infinite 
computation $\norm(t) = e_0 \to e_1 \to \ldots$ in $C_R$.
Every step of $\erasure(A)$ is needed.
\end{corollary}
\begin{proof}
Operation $\norm$ of $C_R$ applied to an expression $t$ of $R$
applies operation $\head$ to every outermost operation-rooted
subexpression of $t$.  All these expressions are needed by
Def.~\ref{needed-node}.
The claim is therefore a direct consequence of Th.~\ref{main}.
\end{proof}

\begin{corollary}
\label{equivalence}
Let $R$ be an inductively sequential system.
For all expressions $t$ and constructor forms $u$ of $R$,
$t \tostar u$ in $R$ if, and only if, $\norm(t) \tostar u$ in $C_R$
modulo a renaming of nodes.
\end{corollary}\begin{proof}
Let $A$ denote some innermost computation of $\norm(t)$.
Observe that if $A$
terminates in a constructor form $u$ of $R$, then every
innermost computation of $\norm(t)$ terminates in $u$
because the order of the reductions is irrelevant.
Therefore, we consider whether $A$ terminates normally.
Case 1: $A$ terminates normally.
If $\norm(t) \tostar u$, then by Lemma \ref{simulated}, point 1,
$t \tostar u$. 
Case 2: $A$ does not terminate normally.
We consider whether $A$ aborts.
Case 2a: $A$ aborts.
Suppose $\norm(t)=e_0 \to e_1 \to \ldots \to e_i$, and
the step of $e_i$ reduces a redex $r$ to ``abort''.
By Theorem~\ref{main}, $r$ is needed for $e_i$, but there is
no rule in $R$ that reduces $r$, hence $t$ has no
constructor form.
Case 2b: $A$ does not terminates.
Every step of $\erasure(A)$ is needed.
The complete tree unraveling \cite[Def. 13.2.9]{Terese03}
of the rules of $R$ and the states of $\erasure(A)$, gives
an orthogonal \emph{term} rewriting system and a
computation of the unraveled $t$.
Since redexes are innermost, in this computation
an infinite number of needed redexes are reduced.  
The hypernormalization of the needed strategy
\cite[Sect. 9.2.2]{Terese03} shows that
hence $t$ has no constructor form.
\end{proof}
\noindent
The object code $C_R$ for a rewrite system $R$ is subjectively simple.
Since innermost reductions suffice for the execution,
operations $\head$ and $\norm$ can be coded as functions
that take their argument by-value. 
This is efficient in most programming languages.
Corollary \ref{equivalence}, in conjunction with Theorem \ref{main},
shows that $C_R$ is a good object code: it produces the value of an
expression $t$ when $t$ has such value, and it produces
this value making only steps that must be made by any rewrite computation.
One could infer that there cannot be a substantially
better object code, but this is not true.
The next section discusses why.

\section{Transformation}
\label{Transformation}

We transform the object code to avoid totally, or partially,
constructing certain contracta.  
The transformation consists of two phases.

The first phase replaces certain rules of $\head$.
Let $r$ be a rule of $\head$ in which
$\head$ is recursively applied to a variable, say $x$,
as in the third rule of (\ref{append-compiled}).
Rule $r$ is replaced by the set $S_r$ of rules obtained as follows.
A rule $r\!_f$ is in $S_r$, iff $r\!_f$ is obtained from $r$
by instantiating $x$ with $f(x_1, \ldots x_n)$, where
$f/n$ is an operation of $R$,
$x_1, \ldots x_n$ are fresh variables,
and the sorts of $f(x_1, \ldots x_n)$ and $x$ are the same.
If a rule in $S_r$ still applies $\head$ to another variable,
it is again replaced in the same way.

\begin{example}
The following rule originates from instantiating
\code{y} for ``\code{++}'' in the third rule of (\ref{append-compiled}).
\equprogram{%
\label{phase-1-example}%
\head([]++(u++v)) = \head(u++v)
}
\end{example}
\vspace*{-3ex}
\noindent
The first phase of the transformation 
ensures that $\head$ is always applied to an
expression rooted by some operation $f$ of $R$.
The second phase introduces, for each operation $f$ of $R$,
a new operation, denoted $\headcomp{f}$.
This operation is the composition of $\head$ with $f$,
and then replaces every occurrence of the composition of $\head$ with $f$
with $\headcomp{f}$.

\begin{example}
The second phase transforms (\ref{phase-1-example}) into:
\equprogram{%
\label{phase-2-example}%
\headcomp{\code{\small ++}}([],u++v) = \headcomp{\code{\small ++}}(u,v)
}
\end{example}
\noindent
After the second phase, operation $\head$ can be eliminated
from the object code since it is no longer invoked.
We denote the transformed $C_R$ with $T_R$ and the outcome
of the first phase on $C_R$ with $C'_R$.
The mapping $\tau$, from expressions of $C_R$ to expressions
of $T_R$, formally defines the transformation:
\begin{equation}
\tau(t) = \left\{
\begin{array}{@{}l l@{}}
\headcomp{f}(\tau(t_1),\ldots \tau(t_n)), \quad 
  & \mbox{if $t=\head(f(t_1,\ldots t_n))$;}\\
s(\tau(t_1),\ldots \tau(t_n)),
  & \mbox{if $t=s(t_1,\ldots t_n)$, with $s$ symbol of $R$;}\\
v, & \mbox{if $t=v$, with $v$ variable.}
\end{array}
\right.
\end{equation}
$T_R$ is more efficient than $C_R$
because, for any operation $f$ of $R$,
the application of \headcomp{f} avoids the allocation of
a node labeled by $f$.
This node is also likely to be pattern matched later.

\begin{example}
Consider the usual length--of--a--list operation:
\equprogram{%
\label{length-source}%
length [] = 0 \\
length (\us:xs) = 1+length xs
}
The compilation of (\ref{length-source}),
where we omit rules irrelevant to the point we are making,
produces:
\equprogram{%
\label{length-compiled}%
\head(length([])) = 0 \\
\head(length(\us{}:xs)) = \head(1+length(xs)) \\
$\cdots$
}
The transformation of (\ref{length-compiled}),
where again we omit rules irrelevant to the point we are making,
produces:
\equprogram{%
\label{length-transformed}%
\headcomp{\code{\small length}}([]) = 0 \\
\headcomp{\code{\small length}}(\us{}:xs)) = \headcomp{\code{\small +}}(1,length(xs)) \\
$\cdots$
}
Below, we show the traces of a portion of the computations of~ 
\code{\norm(length\,[7])} executed by $C_R$ (left) and $T_R$ (right),
where the number \code{7} is an irrelevant value.
The rules of ``\code{+}'' are not shown.
Intuitively, they evaluate the arguments to numbers,
and then perform the addition.

\begin{center}
\renewcommand{\arraystretch}{1.2}
\begin{tabular}{@{} l @{\hspace*{3em}} |  @{\hspace*{3em}} l @{}}
\code{\norm(\head(length\,[7])}
  & \code{\norm(\headcomp{\code{\small length}}([7])} \\
\w $\to$ \code{\norm(\head(1\uline{+}length\,[])}
  & \w $\to$ \code{\norm(\headcomp{\code{\small +}}(1,length\,[])} \\
\w $\to$ \code{\norm(\head(1+\head(length\,[]))}
  & \w $\to$ \code{\norm(\headcomp{\code{\small +}}(1,\headcomp{\code{\small length}}([])))} \\
\w $\to$ \code{\norm(\head(1+0))}
  & \w $\to$ \code{\norm(\headcomp{\code{\small +}}(1,0))} \\
\w $\to$ \code{\norm(1)}
  & \w $\to$ \code{\norm(1)}\\
\w $\to$ \code{1}
  & \w $\to$ \code{1}
\end{tabular}
\end{center}

\noindent
$C_R$ constructs the expression
rooted by the underlined occurrence of ``\code{+}'',
and later pattern matches it.
The same expression
is neither constructed nor pattern matched by $T_R$.
\end{example}

\medskip
\noindent
The transformation increases the size of a program.  
Certain rules are replaced by sets of rules.
The number of rules in a replacing set is the number of constructors
of some type.
A coarse upper bound of the size of the transformed program 
is a constant factor of the size of the original program.
Modern computers have gigabytes of memory.
We believe that the growth in size could become a problem only in
extreme cases, and likely would not be the most serious problem in those
cases.

\section{Transformation Properties}
\label{Transformation Properties}

We show that both phases of the transformation described in the
previous section preserve the object code computations.

\begin{lemma}
Let $R$ be an inductively sequential system.
Every step of $C_R$ is a step of $C'_R$ and vice versa,
modulo a renaming of nodes.
\end{lemma}
\begin{proof}
Every rule of $C'_R$ is an instance of a rule of $C_R$.
Hence every step of $C'_R$ is a step of $C_R$.
For the converse, 
let $t \to u$ be a step of $C_R$ where some rule $r$ is applied.
It suffices to consider the case in which $t$ is the redex and 
the rule $r$ applied in the step is not in $C'_R$.
Let $v$ be the variable ``wrapped'' by $\head$ in $r$.
Rule $r$ is output by statement either 04 or 15 of procedure
\code{compile}.
We show that in both cases the match of $v$, say $s$,
is an operation-rooted subexpression of $t$.
If $r$ is output by statement 04, this
property is ensured by Lemma \ref{operation-rooted}.
If $r$ is output by statement 15, and the match of $v$ were
constructor-rooted, then some rule output by statement
14 of procedure \code{compile}, which textually precedes $r$
and is tried first, would match $t$.
Therefore, let $f/n$ be the root of $s$.
By the definition of phase 1 of the transformation,
rule $r[f(x_1,\ldots x_n)/v]$ is in $C'_R$.
Therefore, modulo a renaming of nodes, $t \to u$ in $C'_R$
\end{proof}

\begin{corollary}
\label{phase-2}
Let $R$ be an inductively sequential system.
For every operation-rooted expression $t$
and head-constructor form $u$ of $R$,
$\head(t) \toplus u$ in $C'_R$ if, and only if,
$\tau(\head(t)) \toplus u$ in $T_R$
modulo a renaming of nodes.
\end{corollary}
\begin{proof}
Preliminarily, we show that for any $s$,
$\head(t) \to s$ in $C'_R$ iff $\tau(\head(t)) \to \tau(s)$ in $T_R$.
Assume $\head(t) \to s$ in $C'_R$.  There exists a rule $l \to r$ of $C'_R$
and a match (graph homomorphism)
$\sigma$ such that $\head(t)=\sigma(l)$ and $s=\sigma(r)$.
From the definition of phase 2 of the transformation,
$\tau(l) \to \tau(r)$ is a rule of $T_R$.
We show that this rule reduces $\tau(\head(t))$ to $\tau(s)$.
Since $\tau$ is the identity on variables,
and $\sigma$ is the identity on non variables,
$\sigma \circ \tau = \tau \circ \sigma$.
Thus $\tau(\head(t)) = \tau(\sigma(l)) = \sigma(\tau(l)) \to
\sigma(\tau(r)) = \tau(\sigma(r)) = \tau(s)$.
The converse is similar because there a bijection between 
the steps of $C'_R$ and $T_R$.

Now, we prove the main claim.  First, the claim just proved
holds also when $\head(t)$ is in a context.
Then, an induction
on the length of $\head(t) \toplus u$ in $C'_R$ 
shows that $\tau(\head(t)) \toplus \tau(u)$ in $T_R$.
Since by assumption $u$ is an expression of $R$, 
by the definition of $\tau$, $\tau(u)=u$.
\end{proof}

\noindent
Finally, we prove that object code and transformed object code 
execute the same computations.
\begin{theorem}
Let $R$ be an inductively sequential system.
For all expressions $t$ and $u$ of $R$,
$\norm(t) \toplus u$ in $C_R$ if, and only if, $\norm(t) \toplus u$ in $T_R$.
\end{theorem}
\begin{proof}
In the computation of $\norm(t)$ in $C_R$, 
by the definition of $\tau$, 
each computation of $\head(s)$ in $C_R$,
for some expression $s$,
is transformed into a computation of $\tau(\head(s))$ in $T_R$.
By Lemma \ref{simulated}, the former ends in a head-constructor
form of $R$.  Hence, by Corollary~\ref{phase-2}, 
$\tau(\head(s))$ ends in the same head-constructor form of $R$.
Thus, $\norm(t) \toplus u$ in $T_R$ produces the same result.
The converse is similar.
\end{proof}

\section{Benchmarking}
\label{Benchmarking}

Our benchmarks use integer values.
To accommodate a built-in integer in a graph node,
we define a kind of node whose
label is a built-in integer rather than a signature symbol.
An arithmetic operation, such as addition,
retrieves the integers labeling its argument nodes,
adds them together, and allocates a new node labeled
by the result of the addition.

Our first benchmark evaluates \code{length($l_1 \code{++}\, l_2$)},
where \code{length} is the operation defined in (\ref{length-source}).
In the table below, we compare the same rewriting computation executed
by $C_R$ and $T_R$.
We measure the number of rewrite and shortcut steps executed, the
number of nodes allocated, and the number of node labels compared by
pattern matching.
The ratio between the execution times of $T_R$ and $C_R$
varies with the implementation language, the order of
execution of some instructions, and other code details that 
would seem irrelevant to the work being performed.
Therefore, we measure quantities that are language and code independent.
The tabular entries are in units per
10 rewrite steps of $C_R$, and are constant functions
of this value except for very short lists.  
For lists of one million elements,
the number of rewrite steps of $C_R$ is two million.

\begin{center}
\renewcommand{\arraystretch}{1.25}
\begin{tabular}{@{} | l || r | r | r | @{}}
\hline
$\emph{length}\,(l_1 \code{++}\, l_2)$ & $C_R$ & $T_R$ & $O_R$ \\
\hline
\hline
rewrite steps     &   10  & 6 & 6 \\
\hline
shortcut steps     &   0  & 4 & 4 \\
\hline
node allocations  &  20  & 16 & 12 \\
\hline
node matches      &  40  & 26 & 18 \\
\hline 
\end{tabular}
\end{center}
\vspace*{.5ex}
\noindent
The column labeled $O_R$ refers to object code that
further shortcuts needed steps using the same idea
behind the transformation.
For example, in the second rule of
(\ref{length-compiled}), both arguments of the addition in
the right-hand side are needed.
This information is known at compile-time,
therefore the compiler can wrap an application of $\head$
around the right operand of ``\code{+}'' in the right-hand side of the rule.
%
\equprogram{%
\label{length-compiled-2}%
\head(length(\us{}:xs)) = \head(1+\head(length(xs)))
}
The composition of $\head$ with \code{length}
is replaced by \headcomp{\code{\small length}}
during the second phase.
The resulting rule is:
\equprogram{%
\label{length-transformed-2}%
\headcomp{\code{\small length}}(\us{}:xs)) 
  = \headcomp{\code{\small +}}(1,\headcomp{\code{\small length}}(xs))
}
Of course, there is no need to allocate a node for
expression \code{1}, the first argument of the addition,
every time rule (\ref{length-transformed}) or 
(\ref{length-transformed-2}) is applied.
A single node can be shared by the entire computation.
However, since the first argument of the application of \headcomp{\code{\small +}}
is constant, this application
can be specialized or partially evaluated as follows:
\equprogram{%
\label{length-transformed-3}%
\headcomp{\code{\small length}}(\us{}:xs)) 
  = \headcomp{\code{\small +1}}(\headcomp{\code{\small length}}(xs))
}
The application of rule
(\ref{length-transformed-3}) allocates no node of the contractum.
In our benchmarks, we ignore any optimization that is not directly related
to shortcutting.
Thus $C_R$, $T_R$ and $O_R$
needlessly allocate this node every time these rules are applied.

The number of shortcut steps of $T_R$ and $O_R$ remain the same
because, loosely speaking, $O_R$ shortcuts a step that was already
shortcut by $T_R$, but the number of nodes allocated and matched
further decreases.
The effectiveness of $T_R$ to reduce node allocations or pattern matching
with respect to $C_R$ varies with the program and the computation.

Our second benchmark computes the $n$-th Fibonacci number
for a relatively large value of $n$.
The program we compile is:
\equprogram{%
\label{fibonacci}%
fib 0 = 0 \\
fib 1 = 1 \\
fib n = fib (n-1) + fib (n-2)
}
To keep the example simple, we assume that
pattern matching is performed by scanning the rules in textual order.
Therefore, the last rule is applied only when the argument of \code{fib}
is neither 0 nor 1.

%
\begin{center}
\renewcommand{\arraystretch}{1.25}
\begin{tabular}{@{} | l || r | r | r | @{}}
\hline
$\emph{fib}\,(n)$ & $C_R$ & $T_R$ & $O_R$ \\
\hline
\hline
rewrite steps     &   10  & 8 & 8 \\
\hline
shortcut steps     &   0  & 2 & 2 \\
\hline
node allocations  &  24  & 22 & 10 \\
\hline
node matches      &  44  & 26 & 16 \\
\hline 
\end{tabular}
\end{center}
\vspace*{.5ex}
The tabular entries are in units per
10 rewrite steps of $C_R$ and are constant functions
of this value except for very small arguments of \code{fib}.
For $n=32$, the number of steps of $C_R$ is about 17.5 million.
With respect to $C_R$, $T_R$ avoids the construction of the
root of the right-hand side of
the third rule of (\ref{fibonacci}).
$O_R$ transforms the right-hand side of this rule into:
\equprogram{%
\label{fibonacci-opt}%
\headcomp{\code{\small +}}(%
  \headcomp{\code{\small fib}}(%
    \headcomp{\code{\small -}}(%
      n,1%
    )%
  ),%
  \headcomp{\code{\small fib}}(%
    \headcomp{\code{\small -}}(%
      n,2%
    )%
  )%
)%
}
since every node that is not labeled by the variable or
the constants \code{1} and \code{2} is needed.
In this benchmark,
there is also no need to allocate a node for
either \code{1} or \code{2} every time
(\ref{fibonacci-opt}) is constructed/executed.
With this further optimization, the step would 
allocate no new node for the contractum,
and the relative gains of our approach
would be even more striking.

\section{Functional Logic Programming}
\label{Functional Logic Programming}

Our work is motivated by the implementation of functional logic languages.
The graph rewriting systems modeling functional logic programs
are a superset of the inductively sequential ones.
A minimal extension consists of a single binary operation,
called \emph{choice}, denoted by the infix symbol ``\code{?}''.
An expression $x$ \code{?} $y$ reduces non-deterministically
to $x$ or $y$.
There are approaches \cite{Antoy11ICLP,AntoyBrownChiang06Termgraph}
for rewriting computations involving the \emph{choice} operation
that produce all the values of an expression without ever
making a non-deterministic choice.
These approaches are ideal candidates to host our compilation scheme.

Popular functional logic languages allow variables,
called \emph{extra variables}, which occur in the right-hand
side of a rewrite rule, but not in the left-hand side.
Computations with extra variables are
executed by \emph{narrowing} instead of rewriting.
Narrowing simplifies encoding certain
programming problems into programs \cite{Antoy10JSC}.
Since our object code selects rules in textual order,
and instantiates some variables of the rewrite system,
narrowing with our object code is not straightforward.
However, there is a technique \cite{AntoyHanus06ICLP}
that transforms a rewrite system \emph{with} extra variables
into an equivalent system \emph{without} extra variables.
Loosely speaking, ``equivalent'', in this context, means a system
with the same input/output relation.
In conjunction with this technique, our compiler generates
code suitable for narrowing computations.

\section{Related Work}
\label{Related Work}

The redexes that we reduce are needed to obtain a constructor-rooted
expression, therefoer they are closely related to the notion
of \emph{root-neededness} of \cite{Middeldorp97POPL}.
However, we are interested only in normal forms that are constructor forms.
In contrast to a computation according to \cite{Middeldorp97POPL},
our object code may abort the computation of an expression $e$
if no constructor normal form of $e$ is reachable, even if $e$
has a needed redex.  This is a very desirable property in 
our intended domain of application
since it saves useless rewrite steps, and in some cases may
lead to the termination of an infinite computation.

Machines for graph reduction have been 
proposed \cite{BurnEtAl88LFP,Kieburtz85} for the implementation
of functional languages.
While there is a commonality of intent, these efforts differ from ours
in two fundamental aspects.
Our object code is easily translated into a low-level language
like $C$ or assembly, whereas these machines have instructions
that resemble those of an interpreter.
There is no explicitly notion of \emph{need} in the computations
performed by these machines.
Optimizations of these machines are directed toward their internal
instructions, rather than the needed steps of a computation by rewriting,
a problem less dependent on any particular mechanism used
to compute a normal form.

Our compilation scheme has similaties with
deforestation \cite{Wadler90TCS},
but is complementary to it.
Both anticipate rule application, to avoid the construction of
expressions that would be quickly taken apart and disposed.  
This occurs when a function producing
one of these expressions is nested within a function
consuming the expression.
However, our expressions are operation-rooted
whereas in deforestation they are constructor-rooted.
These techniques can be used independently
of each other and jointly in the same program.

A compilation scheme similar to ours
is described in \cite{Antoy92CTRS}.
This effort makes no claims of correctness, of executing only needed steps
and of shortcutting needed steps.
Transformations of rewrite systems
for compilation purposes are described in
\cite{FokkinkDePol97MFCS,KampermanWalters96CWI}.
These efforts are more operational than ours.
A compilation with the same intent as ours
is described in \cite{AntoyJost13PPDP}.
The compilation scheme is different.
This effort does not claim to execute only needed steps,
though it shortcuts some of them. Shortcutting is obtained
by defining \emph{ad-hoc} functions whereas
we present a formal systematic way through specializations
of the \emph{head} function.

\section{Conclusion}
\label{Conclusion}

Our work addresses rewriting computations for the implementation
of functional logic languages.  We presented two major results.

The first result is a compilation scheme for inductively sequential
graph rewriting systems.  The object code generated by our scheme
has very desirable properties: it is simple consisting of only
two functions that take arguments by value, it is theoretically
efficient by only executing needed steps, and it is complete
in that it produces the value, when it exists, of any expression.
The two functions of the object code are easily generated
from the signature of the rewrite system and 
a traversal of the definitional trees of its operations.

The second result is a transformation of the object code that
shortcuts some rewrite steps.  Shortcutting avoids partial or
total construction of the contractum of a step by composing
one function of the object code with one operation symbol of
the rewrite system signature.  This avoids the construction of
a node and in some cases and its subsequent pattern matching.
Benchmarks show that the savings in node allocation and matching
can be substantial.

Future work will rigorously investigate the extension of
our compilation technique to rewrite systems with
the \emph{choice} operation and extra variables, as
discussed in Sect.~\ref{Functional Logic Programming},
as well as systematic opportunities to shortcut needed steps
in situations similar to that discussed in Sect.~\ref{Benchmarking}.

\subsection*{Acknowledgments}

This material is based upon work partially supported by the National
Science Foundation under Grant No. CCF-1317249.
This work was carried out while the second author was visiting
the University of Oregon. The second author wishes to thank Zena
Ariola for hosting this visit.
The authors wish to thank Olivier Danvy for insightful comments
and the anonymous reviewers for their careful reviews.

\bibliographystyle{eptcs}

\end{document}